\documentclass{article}
\usepackage{graphicx} 
\usepackage{amsthm}
\usepackage{amsmath}
\usepackage{amsfonts}
\usepackage{mathrsfs}
\usepackage{amssymb}
\usepackage{cancel}\usepackage{authblk}
\usepackage[hidelinks]{hyperref} 
\usepackage{simplewick}
\usepackage{tikz}
\usepackage{circuitikz}
\usetikzlibrary{ decorations.markings}
\usetikzlibrary{cd}
\usepackage{mathtools}

\usepackage[nameinlink, capitalise]{cleveref} 
\usepackage[numbers,compress]{natbib}

\newtheorem{theorem}{Theorem}[section]
\theoremstyle{definition}
\newtheorem{definition}[theorem]{Definition}
\newtheorem{lemma}[theorem]{Lemma}

\newtheorem{example}[theorem]{Example}
\theoremstyle{remark}

\title{The causal set reduction formula}

\author{Stav Zalel}
\affil{Department of Applied Mathematics and Theoretical Physics,
University of Cambridge, CB3 0WA, United Kingdom \\ and Homerton College, University of Cambridge, CB2 8PQ, United Kingdom}
\begin{document}

\maketitle

\begin{abstract}
    We derive a reduction formula for matrix elements on a causal set background. We derive an infinite tower of relations between correlators, akin to the Schwinger–Dyson equations of the continuum. Combining these two results we are able to express matrix elements in three different forms: as a path integral and as two distinct sums of correlators. We sketch the form that our method---which circumvents explicit use of differential equation of motion operators---takes in flat continuum spacetime where it provides an alternative expression for the standard LSZ result.
\end{abstract}

\tableofcontents
\newpage
\section{Introduction}
The LSZ reduction formula \cite{LSZ,Lehmann:1957zz} provides an explicit connection between S-matrix elements and time-ordered correlation functions. It is both a tool for perturbative, diagrammatic computation of measurable quantities and for the formal study of quantum field theory (QFT) through the S-matrix programme where it is used to derive analytic properties of scattering amplitudes \cite{Eden:1966dnq}. Notably, with regards to computation of S-matrix elements, the LSZ formula further cemented the transition from canonical, Hamiltonian methods to covariant methods based on the Lagrangian \cite{Bjorken:1964, Weinberg:1995mt}.

Lagrangian (or path integral) approaches are particularly well-suited to quantum gravity and quantum cosmology \cite{Sorkin:1997gi}. Causal set theory is an approach to quantum gravity in the spirit of the path integral where spacetime takes the form of a discrete Lorentzian substratum, a partial order whose elements are indivisible pieces of spacetime \cite{bombelli1987space,Dowker:2013dog,surya2019causal}. The lightcone structure is encoded by the partial order, a causal order that determines for any pair of spacetime events whether one happened before the other or whether the two are spacelike. The discreteness of the theory is inherent in the condition that the number of spacetime events (that is, the number of indivisible spacetime pieces) contained in any causal interval is finite.

Given the guiding principles and practical constraints of causal set theory, it seems natural (if not essential) to formulate QFT on a causal set via a path integral. The first formulation of free scalar QFT on a causal set, however, was in the language of operators and state-vectors \cite{Johnston:2008za,Johnston:2009fr,Sorkin:2017fcp}. Known as the Sorkin-Johnston formalism, and adapted to the continuum in \cite{Sorkin:2017fcp,Afshordi:2012jf, afshordi2012ground,Zhu:2022kcf}, this approach shuffles the order in which the ingredients of a QFT usually appear in canonical quantisation. In the Sorkin-Johnston formalism, the starting point is the retarded Green function, from which the causal propagator and then the Wightman function are obtained, fixing the free theory via Wick's theorem. Finally, a Fock space and field operators are obtained by diagonalising the Wightman function. Indeed, the only use of the equation of motion in the continuum version of this formalism is to obtain the retarded Green function. On a causal set, the Sorkin-Johnston formalism enables one to make headway without an equation of motion at all by writing down a Green function according to one of the various prescriptions in the literature \cite{Johnston:2008za,X:2017jal,Dable-Heath:2019sej, Aslanbeigi:2014zva,Kastrati:2025etw}---a necessary step, since the discrete, Lorentzian structure of a causal set does not readily furnish differential operators.

Subsequently, the path integral form for free scalar QFT on a causal set was written down in \cite{Sorkin:2011pn}. There, the Wightman function (together with the field operators and vacuum state obtained from it) were used to write down a decoherence functional (a bi-measure, in the spirit of the in-in formalism) from which all correlation functions can be obtained by integration. That paper concluded with a proposal for a path integral for \textit{interacting} QFTs on a causal set, realised in \cite{Albertini:2024srq, Jubb2024}.

To summarise the state-of-the-art \cite{Albertini:2024srq}: in-in and in-out correlators have been defined, both via a canonical operator approach and a path integral\footnote{A formalism for algebraic QFT on causal sets has been separately proposed in \cite{Dable-Heath:2019sej}.}; their perturbative diagrammatic expansions have been derived; and S-matrix elements have been defined using the canonical operator approach---but no path integral formulation for S-matrix elements was known.

Here, we fill this gap by deriving a path integral expression for matrix elements on a causal set. We further derive quantum equations of motion for correlation functions, enabling us to express matrix elements as a sum of correlation functions in a way that resembles the LSZ reduction formula in the continuum.

This paper is structured as follows. Section \ref{sec:prem} provides a brief technical introduction. In Section \ref{sec:integral} we evaluate our central integral, showing that it is equal to a certain sum of correlators. In Section \ref{sec:reduction} we use this integral to give a path integral expression for matrix elements---the causal set reduction formula. In Section \ref{sec:DS} we find that, on a causal set, correlators satisfy an infinite tower of relations, akin to the Schwinger–Dyson equations of the continuum. We give an explicit example of this relation, illustrating how it gives us an additional way to express matrix elements on a causal set. We conclude in Section \ref{sec:discussion} with a comparison with the continuum.

\section{Preliminaries}\label{sec:prem}
We will be concerned with real scalar QFTs on a causal set. The spacetime background on which our QFTs live is a finite, naturally labeled causal set, namely: a partial order on the set $\{1,\ldots,N\}$ for some natural number $N>0$ such that if $p\prec q$ then $p<q$. In other words, the partial order is compatible with the total order of the natural numbers. The integers $1,\ldots, N$ label the spacetime points and can be thought of as our coordinates. The partial order encodes the lightcone structure, where $p\prec q$ means that $p$ is in the causal past of $q$. 

A free QFT is fully determined by its vacuum two-point function, which in our setting takes the form of a complex Hermitian $N\times N$ matrix, $W_{pq}$. We will assume that $W_{pq}$ has been obtained via the Sorkin-Johnston prescription \cite{Sorkin:2017fcp} and can therefore be expressed in operator form as $W_{pq}= \langle 0| \phi_p\phi_q|0\rangle$, where $|0\rangle$ denotes the Sorkin-Johnston vacuum and $\phi_p$ for $p=1,\ldots, N$ are the associated field operators. The field operators can be expanded as,\footnote{The bar over $v^{\lambda}$ and $W$ denotes complex conjugation.}
\begin{equation}\label{free field mode expansion}
    \phi_p=\sum_{\lambda>0} \sqrt{\lambda}\bigg( v^{\lambda}_p a_{\lambda}+\bar{v}^{\lambda}_p a_{\lambda}^{\dagger}\bigg),
\end{equation}
where $v^{\lambda}$ is an eigenvector with eigenvalue $\lambda$ of the causal propagator $W-\overline{W}$,
\begin{equation}
    (W_{pq}-\overline{W}_{pq})\ v_q^{\lambda}=\lambda  v_p^{\lambda},
\end{equation}
and $a_{\lambda},a_{\lambda}^{\dagger}$ are ladder operators satisfying $[a_{\lambda},a_{\mu}^{\dagger}]=\delta_{\lambda,\mu}$ and $a_{\lambda}|0\rangle=0$ for all $\lambda,\mu>0$. The positive eigenvalues and their associated eigenvectors play the role of positive frequencies and mode functions, respectively. The Fock space is built via applications of the creation operators on the vacuum, \newline $|\lambda_1\ldots \lambda_k\rangle=\frac{1}{\sqrt{\lambda_1\ldots \lambda_k}}a_{\lambda_1}^{\dagger}\ldots a_{\lambda_k}^{\dagger}|0\rangle$. The Feynman propagator $\Delta^F$ is defined as,
\begin{equation}
    \Delta^F_{pq}=
    \begin{cases}
     W_{pq} & \text{if $p\succ q$}\\
       W_{qp} & \text{ otherwise},\\
    \end{cases}\end{equation}
where $W_{pq}=W_{qp}$ if $p$ and $q$ are spacelike. Equivalently, $\Delta^F$ is the causally-ordered 2-point function, where causal ordering (our analogue of the continuum's time ordering) orders products of field operators such that if $p\succ q$ then $\phi_p$ is to the left of $\phi_q$. We write $C[\phi_p\phi_q]$ to denote the causal ordering of $\phi_p\phi_q$. 

On top of this free construction, we now add local polynomial interactions in the field. Given a local interaction $V_p$ (a local polymial in $\phi_p$) at all points $p=1,\ldots, N$, the evolution operator $U_{p}$ for $p>1$ is given by the causally ordered product,
\begin{equation}
     U_{p}=C\bigg[\prod_{1\leq  q < p} e^{-iV_q}\bigg].
\end{equation}
Given a local polynomial in the free field, $F(\phi_p)$, the Heisenberg operator, $F^H(\phi_p)$, is defined as,\footnote{When $p=1$, $U_1=1$ and $F^H(\phi_1)=F(\phi_1)$.}
\begin{equation}
     F^H(\phi_p)=U_{p}^{\dagger} F(\phi_p) U_p.
\end{equation}
The S-operator is defined as,
\begin{equation}
     S=C\bigg[\prod_{q=1}^N e^{-iV_q}\bigg].
\end{equation}
The interacting vacuua are defined as $|0\rangle_{in}=|0\rangle$ and $|0\rangle_{out}=S^{\dagger}|0\rangle$, and similarly for particle states, $|\lambda_1\ldots \lambda_k\rangle_{in}=|\lambda_1\ldots \lambda_k\rangle$ and $|\lambda_1\ldots \lambda_k\rangle_{out}=S^{\dagger}|\lambda_1\ldots \lambda_k\rangle$. Due to our choice of vacuum, the overlap of a pair of particle states in the free theory is non-vanishing only if the two states are identical. In the interacting theory, the overlap of a pair of particles is always non-vanishing since there are no spacetime symmetries to dictate transition rules (for instance, no analogue of momentum conservation). Therefore, particle production can come about in the interacting (but not in the free) theory.

There are four objects that will be of most use to us and we define the following shorthand for them.
Let $F=F(\phi_1,\ldots,\phi_N)$ denote a polynomial operator in the free field operators. Then,
\begin{align}
      &\langle F\rangle_{\text{free
      }}:=  \langle 0| C[F]|0\rangle, \label{def:free} \\
       &\langle F\rangle\ \ \ \ \ := \  \frac{_{out}\langle 0| C_H[F]|0\rangle_{in}}{_{out}\langle 0|0\rangle_{in}}\label{def:int},\\
       &\langle \lambda_1\ldots \lambda_k|F|\mu_1\ldots \mu_l\rangle_{\text{free
      }}:=\langle \lambda_1\ldots \lambda_k|C[F]|\mu_1\ldots \mu_l\rangle,\label{def:freemat}\\
       &\langle \lambda_1\ldots \lambda_k|F|\mu_1\ldots \mu_l\rangle\ \ \ \ \ :=\ \frac{ _{out}\langle \lambda_1\ldots \lambda_k|C_H[F]|\mu_1\ldots \mu_l\rangle_{in}}{_{out}\langle 0|0\rangle_{in}},\label{def:intmat}
\end{align}
where $C_H[F]$ denotes the operator obtained by applying causal ordering to $F$ and then replacing every free field with a Heisenberg field. For instance, $C_H[\phi_2\phi_5^2]=(\phi_5^H)^2\phi_2^H$.
These can be computed via Wick's theorem where the contractions are given by $ \contraction{}{\phi_p}{}{\phi_q}\phi_p\phi_q=\Delta^F_{pq}$ and $ \contraction{}{\phi_p}{}{|\lambda\rangle}\phi_p| \lambda\rangle=v_p^{\lambda}$. 

To conclude this section, we review the path integral expression for \eqref{def:free} and \eqref{def:int}. In Section \ref{sec:reduction}  we derive the path integral expression for \eqref{def:intmat} (Theorem \ref{th:lsz}), of which the path integral expressions for \eqref{def:free}-\eqref{def:freemat} are special cases. 

Given a polynomial operator $F=F(\phi_1,\ldots,\phi_N)$,  its \emph{corresponding function} $f=f(\xi_1,\ldots,\xi_N)$ is a function whose dependence on $\xi_i$ is the same as the dependence of $F$ on $\phi_i$, where $\xi=(\xi_1,\ldots,\xi_N)\in\mathbb{R}^N$ is a real scalar field history or configuration. For instance, the corresponding function of $\phi_1^2\phi_3$ is $\xi_1^2\xi_3$.\footnote{ Although formally the map from polynomial operators to their corresponding functions is many-to-one (since the $\xi_p$'s commute while the $\phi_p$'s do not) for us the map will effectively  be a bijection because our operator $F$ will always appear within a causally ordered correlator $\langle F\rangle$.} In particular, write $\mathscr{V}$ to denote the corresponding function for the interaction potential $V=\sum_{p=1}^NV_p$ and $\mathscr{S}$ to denote the corresponding function for the $S$-operator, hence $\mathscr{S}=e^{-i\mathscr{V}}$.
 
When the Feynman propagator matrix $\Delta^F$ is invertible, the path integral expression for in-out correlators is given by \cite{Albertini:2024srq},\footnote{The path integral form of $\langle F\rangle_{\text{free}}$ is obtained from \eqref{eq:pthint} by setting $\mathscr{V}=0$ and noting that $\int d^N\xi \ D=1$.}
\begin{equation}\label{eq:pthint}
\begin{split}
 &\langle  F\rangle
=\frac{\int d^N\xi\  f\mathscr{S}D }{\int d^N\xi\ \mathscr{S} D },
\end{split}
\end{equation}
where the measure $D\equiv D(\xi_1,\ldots,\xi_N)$ takes the form \cite{Sorkin:2011pn},\footnote{For further discussion of the form of the measure see also \cite{
Albertini:2024srq,Jubb2024}.}
\begin{equation}\label{eq:dcf0}
    \begin{split}
    D(\xi_1,\ldots,\xi_N) &= \left\langle\delta(\phi_N-\xi_N)\cdots\delta(\phi_2-\xi_2)\delta(\phi_1-\xi_1)\right\rangle_{\text{free}}\\
    &=\int \frac{d\omega^1}{2\pi} \ldots \frac{d\omega^N}{2\pi} \ e^{-i\omega\cdot\xi-\frac{1}{2}\omega \cdot\Delta^F \cdot\omega } \\
&=\sqrt{\frac{det\mathcal{E}}{(2\pi)^N}}e^{-\frac{1}{2}\xi \cdot {\mathcal{E}} \cdot \xi }. 
    \end{split}
\end{equation}
where $\mathcal{E}\equiv{\Delta^F}^{^{-1}}$ is the inverse of the Feynman propagator matrix. Throughout this work, we assume that $\mathcal{E}$ exists and that $\xi \cdot Re[\mathcal{E}] \cdot \xi\geq 0$ for all $\xi\in\mathbb{R}^N$. For ease of notation we will use the shorthand $\mathcal{E}_k\hspace{-0.5mm} \cdot \hspace{-0.5mm}\phi\equiv \mathcal{E}_{ka}\phi_a$ and $\mathcal{E}_k\hspace{-0.5mm} \cdot \hspace{-0.5mm}\xi\equiv \mathcal{E}_{ka}\xi_a$.

\section{Evaluating the integral $I_{\boldsymbol{\alpha}}(f)$}\label{sec:integral}

This section is dedicated to defining and evaluating the integral $I_{\boldsymbol{\alpha}}(f)$ which we will need for stating and proving our theorems in the subsequent sections.

We begin with some notation. Let $\boldsymbol{\alpha}=(\alpha^1,\ldots,\alpha^N)$ and $\boldsymbol{\beta}=(\beta^1,\ldots,\beta^N)$ be two ordered lists of non-negative integers, \textit{i.e.} $\boldsymbol{\alpha},\boldsymbol{\beta}\in\mathbb{N}^N$. The sum and difference of lists is defined component-wise as $\boldsymbol{\alpha}\pm\boldsymbol{\beta}=(\alpha^1\pm\beta^1,\ldots, \alpha^N\pm\beta^N)$. The size of a list $\boldsymbol{\alpha}$ is $\alpha=\sum_{i=1}^N\alpha^i$. 
We write $\boldsymbol{\alpha}\leq \boldsymbol{\beta}$ if $\alpha^i\leq \beta^i \ \forall \ i$. The zero list in $\mathbb{N}^N$ is $\boldsymbol{0}=(\underbrace{0,\ldots,0}_{N})$.

The factorial of a list $\boldsymbol{\alpha}$ is $\boldsymbol{\alpha}!=\prod_i\alpha^i!$ and the binomial coefficient of two lists is given by $\binom{\boldsymbol{\alpha}}{\boldsymbol{\beta}}=\frac{\boldsymbol{\alpha}!}{\boldsymbol{\beta}!(\boldsymbol{\alpha}-\boldsymbol{\beta})!}=\prod_{i}\binom{\alpha^i}{\beta^i}$. The multinomial of lists is given by  $\binom{\boldsymbol{\alpha}}{\boldsymbol{\beta_0},\boldsymbol{\beta_1},\ldots,\boldsymbol{\beta_n}}=\frac{\boldsymbol{\alpha}!}{\prod_i\boldsymbol{\beta_i}!}$.

A partition $p$ of $\boldsymbol{\alpha}$ into $n$ parts is a multiset of non-vanishing lists, that is a pair $(\mathbb{N}^N\setminus\{\boldsymbol{0}\},m)$ with $m:\mathbb{N}^N\setminus\{\boldsymbol{0}\}\longrightarrow \mathbb{N}$, such that $\sum_{\boldsymbol{\beta}}m(\boldsymbol{\beta})\cdot \boldsymbol{\beta}=\boldsymbol{\alpha}$ and $\sum_{\boldsymbol{\beta}}m(\boldsymbol{\beta})=n$. We write $P_n(\boldsymbol{\alpha})$ to denote the set of partitions of $\boldsymbol{\alpha}$ into $n$ parts. Given a partition $p\in P_n(\boldsymbol{\alpha})$ and an integer $j\in\{1,\ldots,N\}$, we define the set,
\begin{equation}\label{def set I}
    \mathcal{I}_j:=\{\boldsymbol{\beta}\in\mathbb{N}^N | \beta^j>0 \text{ and } m(\boldsymbol{\beta})>0\}.
\end{equation}
The \emph{permutation factor} of a partition $p\in P_n(\boldsymbol{\alpha})$ is $\varpi(p)=\prod_{\boldsymbol{\beta}}m(\boldsymbol{\beta})!\ $. It will sometimes be helpful to label the elements of a partition $p\in P_n(\boldsymbol{\alpha})$, and we will write $p=\{\boldsymbol{\beta_1},\ldots \boldsymbol{\beta_n}\}$. It does not matter which element is $\boldsymbol{\beta_1}$, which is $\boldsymbol{\beta_2},\ldots$ in this sense the partitions are unordered. Whenever we write,
\begin{equation}
\sum_{\substack{P_n(\boldsymbol{\alpha})}}\hspace{-2.5mm}\ g(\boldsymbol{\beta_1},\ldots,\boldsymbol{\beta_n}),
\end{equation}
it is implied that $\{\boldsymbol{\beta_1},\ldots,\boldsymbol{\beta_n}\}\in P_n(\boldsymbol{\alpha}) $. Given a partition $p=\{\boldsymbol{\beta_1},\ldots \boldsymbol{\beta_n}\}$ and some $\boldsymbol{\beta}$ with $m(\boldsymbol{\beta})>0$ we write $\underbrace{\boldsymbol{\beta_1},\ldots, \boldsymbol{\beta_n}}_{\cancel{\boldsymbol{\beta}}}$ to denote the collection of $n-1$ lists constructed from $\boldsymbol{\beta_1},\ldots, \boldsymbol{\beta_n}$ by the removal for a single occurrence of $\boldsymbol{\beta}$.

We denote multivariable partial derivatives of functions and operators, respectively, as follows: 
\begin{equation}
    \begin{split}
&\partial^{\boldsymbol{\alpha}}f=\frac{\partial^{\alpha^1}}{\partial\xi_1}\dots \frac{\partial^{\alpha^N}}{\partial\xi_N}f,\\
&\partial^{\boldsymbol{\alpha}}F=\frac{\partial^{\alpha^1}}{\partial\phi_1}\dots \frac{\partial^{\alpha^N}}{\partial\phi_N}F,
    \end{split}
\end{equation}
with $\partial^{\boldsymbol{0}}=1$. The general Leibniz rule in this notation is,
\begin{equation}
    \partial^{\boldsymbol{\beta}}(fg)=\sum_{\boldsymbol{0}\leq \boldsymbol{\gamma}\leq \boldsymbol{\beta}} \binom{\boldsymbol{\beta}}{\boldsymbol{\beta}-\boldsymbol{\gamma}}\partial^{\boldsymbol{\gamma}}f\partial^{\boldsymbol{\beta}-\boldsymbol{\gamma}}g.
\end{equation}

Now, we define and evaluate the integral $I_{\boldsymbol{\alpha}}(f)$.

\begin{definition}[The integral $I_{\boldsymbol{\alpha}}(f)$]\label{definitionI}
    For any $\boldsymbol{\alpha}\in\mathbb{N}^N$ and function $f=f(\xi_1,\ldots,\xi_N)$, the integral $I_{\boldsymbol{\alpha}}(f)$ is defined as,
    \begin{equation}\label{eq:integral1}
I_{\boldsymbol{\alpha}}(f):=\frac{(-)^\alpha}{\mathcal{N}}\int d^N\xi  \ (f \mathscr{S} \partial^{\boldsymbol{\alpha}}D),
    \end{equation} with $\mathcal{N}\equiv\int d^N\xi\ \mathscr{S} D$.
\end{definition}

\begin{lemma}\label{lemma1}
 Given any list $\boldsymbol{\alpha}\in\mathbb{N}^N$ and polynomial operator $F=F(\phi_1,\ldots,\phi_N)$ with corresponding function $f=f(\xi_1,\ldots,\xi_N)$, and assuming that $\xi \cdot Re[\mathcal{E}] \cdot \xi\geq 0$ for all $\xi\in\mathbb{R}^N$,
 then,   
      \begin{equation}\label{eq:lemma1}
    \begin{split}
I_{\boldsymbol{\alpha}}(f)=\big\langle \partial^{\boldsymbol{\alpha}}F\big\rangle+\sum_{n=1}^\alpha(-i)^n\hspace{-2.5mm}\sum_{\substack{\boldsymbol{\beta_0}< \boldsymbol{\alpha} \\ \alpha-\beta_0\geq n}}\hspace{-0.5mm}\sum_{\substack{P_n(\boldsymbol{\alpha}-\boldsymbol{\beta_0})}}\hspace{-2.5mm}C_{\boldsymbol{\alpha},n}(\boldsymbol{\beta_0},\ldots,\boldsymbol{\beta_n})\ \big\langle \partial^{\boldsymbol{\beta_0}}F\partial^{\boldsymbol{\beta_1}}V\cdots\partial^{\boldsymbol{\beta_n}} V\big\rangle,\\
\end{split}
    \end{equation}
    where it is understood that the sum vanishes when $\alpha=0$ and, 
\begin{equation}\label{eq:lemma_1_coeff}
  C_{\boldsymbol{\alpha},n}(\boldsymbol{\beta_0},\ldots,\boldsymbol{\beta_n})= \frac{1}{\varpi(p)} \binom{\boldsymbol{\alpha}}{\boldsymbol{\beta_0},\boldsymbol{\beta_1},\ldots,\boldsymbol{\beta_n}}.
\end{equation}
\end{lemma}

\begin{proof}
When $\alpha=0$, the result follows directly from \eqref{eq:pthint}. Now consider an arbitrary list  $\boldsymbol{\alpha}>0$ and suppose the lemma holds for all lists smaller than $\boldsymbol{\alpha}$.    
Choose some $j\in\{1,\ldots,N\}$ for which $\alpha^j\not=0$. Define a new list $\boldsymbol{\Delta}$ via $\Delta^i=\delta_{ij}$ for all $i=1,\ldots,N$. Integrate \eqref{eq:integral1} by parts with respect to $\boldsymbol{\Delta}$ to obtain,
\begin{equation}\label{I relation}\begin{split}
          I_{\boldsymbol{\alpha}}(f)= I_{\boldsymbol{\alpha}-\boldsymbol{\Delta}}(\partial^{\boldsymbol{\Delta}}f)-iI_{\boldsymbol{\alpha}-\boldsymbol{\Delta}}( f\partial^{\boldsymbol{\Delta}} H)+\frac{(-)^\alpha}{\mathcal{N}}\int d^N\xi  \ \partial^{\boldsymbol{\Delta}}(f \mathscr{S} \partial^{\boldsymbol{\alpha}-\boldsymbol{\Delta}}D),
     \end{split}
 \end{equation}
and note that the boundary term vanishes under our assumption that $Re[\mathcal{E}]$ is positive definite.
Now, apply the inductive assumption to evaluate the RHS of \eqref{I relation} order by order in $-i$.

First, consider $I_{\boldsymbol{\alpha}-\boldsymbol{\Delta}}(\partial^{\boldsymbol{\Delta}}f)$ and relabel $\boldsymbol{\beta_0}+\boldsymbol{\Delta}\rightarrow\boldsymbol{\beta_0}$ to obtain,
 \begin{equation}\label{eq:first term}
    \begin{split}
I_{\boldsymbol{\alpha}-\boldsymbol{\Delta}}(\partial^{\boldsymbol{\Delta}}f)=\big\langle \partial^{\boldsymbol{\alpha}}F\big\rangle+\sum_{n=1}^{\alpha-1}(-i)^n\hspace{-2.5mm}\sum_{\substack{\boldsymbol{\beta_0}< \boldsymbol{\alpha} \\ \alpha-\beta_0\geq n\\\beta_0^j\geq 1}}\hspace{-0.5mm}\sum_{\substack{P_n(\boldsymbol{\alpha}-\boldsymbol{\beta_0})}}\hspace{-2.5mm}C_{\boldsymbol{\alpha},n}&(\boldsymbol{\beta_0}-\boldsymbol{\Delta},\boldsymbol{\beta_1},\ldots,\boldsymbol{\beta_n})\\
&\big\langle \partial^{\boldsymbol{\beta_0}}F\partial^{\boldsymbol{\beta_1}}V\cdots\partial^{\boldsymbol{\beta_n}} V\big\rangle,\\
\end{split}
    \end{equation}
We find that the zeroth order term of $I_{\boldsymbol{\alpha}}(f)$ is given by the first term on the RHS of \eqref{eq:first term}. To obtain the first order term of $I_{\boldsymbol{\alpha}}(f)$, sum the first order term on the RHS of \eqref{eq:first term} together with $-i\langle \partial^{\boldsymbol{\alpha-\Delta}}(F\partial^{\boldsymbol{\Delta}} H)\rangle\subset -iI_{\boldsymbol{\alpha}-\boldsymbol{\Delta}}( f\partial^{\boldsymbol{\Delta}} H)$ to obtain,

\begin{equation}\begin{split}\label{eq:recursion_n_1}
       C_{\boldsymbol{\alpha},1}(\boldsymbol{\beta_0},\boldsymbol{\alpha}-\boldsymbol{\beta_0})=(1-\delta_{0,\beta_0^j})C_{\boldsymbol{\alpha}-\boldsymbol{\Delta},1}(\boldsymbol{\beta_0}-\boldsymbol{\Delta},\boldsymbol{\alpha}-\boldsymbol{\beta_0})+(1-\delta_{\alpha^j,\beta_0^j})\binom{\boldsymbol{\alpha}-\boldsymbol{\Delta}}{\boldsymbol{\beta_0}}.
     \end{split}
 \end{equation}
When $n=1$, the coefficients \eqref{eq:lemma_1_coeff} simplify to $C_{\boldsymbol{\alpha,1}}(\boldsymbol{\beta_0,\alpha-\beta_0})=\binom{\boldsymbol{\alpha}}{\boldsymbol{\beta_0}}$. Plugging this into the RHS of \eqref{eq:recursion_n_1} and checking the base case completes the result for $n=1$.

Now consider the $n$ order term in $-iI_{\boldsymbol{\alpha}-\boldsymbol{\Delta}}( f\partial^{\boldsymbol{\Delta}} H)$ when $2\leq n\leq \alpha$. After applying the Leibniz rule, this term is given by,
 \begin{equation}\label{eq:nterm}
     \begin{split}
         (-i)^n\sum_{\substack{\boldsymbol{\beta_0}<\boldsymbol{\alpha}-\boldsymbol{\Delta} \\ \alpha-\beta_0\geq n}}\ \sum_{P_{n-1}(\boldsymbol{\alpha}-\boldsymbol{\Delta}-\boldsymbol{\beta_0})} &\sum_{0\leq \boldsymbol{\gamma}\leq \boldsymbol{\beta_0}}C_{\boldsymbol{\alpha}-\boldsymbol{\Delta},{n-1}}(\boldsymbol{\beta_0,\ldots,\beta_{n-1}})\binom{\boldsymbol{\beta_0}}{\boldsymbol{\beta_0-\gamma}}\\
         &\hspace{20mm}\big\langle \partial^{\boldsymbol{\gamma}}F\ \partial^{\boldsymbol{\beta_0-\gamma}+\boldsymbol{\Delta}}V\ \partial^{\boldsymbol{\beta_1}}V\ldots\partial^{\boldsymbol{\beta_{n-1}}}V\big\rangle.
     \end{split}
 \end{equation}
  Relabel $\boldsymbol{\beta_0}\leftrightarrow\boldsymbol{\gamma}$, change the order of summation using,
 \begin{equation}
\sum_{\substack{\boldsymbol{\gamma}<\boldsymbol{\alpha}-\boldsymbol{\Delta}\\ \alpha-\gamma\geq n}} \ \ \sum_{\boldsymbol{0}\leq \boldsymbol{\boldsymbol{\beta_0}}\leq \boldsymbol{\gamma}}=\sum_{\substack{\boldsymbol{\beta_0}<\boldsymbol{\alpha}-\boldsymbol{\Delta}\\ \alpha-\beta_0\geq n}}\ \ \sum_{\substack{\boldsymbol{\beta_0}\leq \boldsymbol{\gamma}<\boldsymbol{\alpha}-\boldsymbol{\Delta}\\ \alpha-\Delta-\gamma\geq {n-1}}}\ ,\end{equation}
and rewrite the domain of the first sum on the RHS as $\boldsymbol{\beta_0}<\boldsymbol{\alpha}$ with the additional constraints $\beta_0^j<\alpha^j$ and $\alpha-\beta_0\geq n$. Define $\boldsymbol{\beta_n}:=\boldsymbol{\gamma}-\boldsymbol{\beta_0}+\boldsymbol{\Delta}$ and use this to eliminate $\boldsymbol{\gamma}$. Finally, noting that
for any function $g(\boldsymbol{\beta_{1},\ldots,\beta_{n-1};\boldsymbol{\beta}})$, we have, \begin{equation}
\sum_{\substack{\boldsymbol{\Delta}\leq \boldsymbol{\beta}<\boldsymbol{\alpha}-\boldsymbol{\beta_0}\\ \beta\leq \alpha-\beta_0-{n-1}}}\ \sum_{P_{n-1}(\boldsymbol{\alpha}-\boldsymbol{\beta_0}-\boldsymbol{\beta)}} g(\boldsymbol{\beta_{1},\ldots,\beta_{n-1};\boldsymbol{\beta}})=
\sum_{P_n(\boldsymbol{\alpha}-\boldsymbol{\beta_0})} \ \sum_{\boldsymbol{\beta \in\mathcal{I}_j}}g(\underbrace{\boldsymbol{\beta_1},\ldots, \boldsymbol{\beta_{n}}}_{\cancel{\boldsymbol{\beta}}};\boldsymbol{\beta})
\end{equation}
where $\mathcal{I}_j$ is as defined in \eqref{def set I}
we can rewrite equation \eqref{eq:nterm} as,
 \begin{equation}\label{eq:second term}
 \begin{split}
&(-i)^n\sum_{\substack{\boldsymbol{\beta_0}<\boldsymbol{\alpha}\\ \beta_0^j<\alpha^j\\ \alpha-\beta_0\geq n}}\ \ \sum_{P_n(\boldsymbol{\alpha}-\boldsymbol{\beta_0})} \ \sum_{\boldsymbol{\beta \in\mathcal{I}_j}}\ \Bigg[C_{\boldsymbol{\alpha}-\boldsymbol{\Delta},{n-1}}(\boldsymbol{\beta_0-\Delta+\beta,\underbrace{\beta_1,\ldots, \beta_n}_{\cancel{\beta}}})\\
     &\hspace{55mm}
\binom{\boldsymbol{\beta_0-\Delta+\beta}}{\boldsymbol{\beta_0}}\big\langle \partial^{\boldsymbol{\beta_0}}F\ \partial^{\boldsymbol{\beta_1}}V\ \partial^{\boldsymbol{\beta_{n-1}}}V\ldots\partial^{\boldsymbol{\beta_n}}V\big\rangle \Bigg].\\
\end{split}
\end{equation} The recursion relation for $2\leq n\leq \alpha-1$ is obtained by summing \eqref{eq:second term} with the order $n$ term in \eqref{eq:first term}. The result is,
\begin{equation}\label{eq:coeff_n_gen}
\begin{split}
    & C_{\boldsymbol{\alpha},n}(\boldsymbol{\beta_0},\ldots, \boldsymbol{\beta_n})\\
    =&(1-\delta_{0,\beta_0^j})C_{\boldsymbol{\alpha}-\boldsymbol{\Delta},n}(\boldsymbol{\beta_0}-\boldsymbol{\Delta},\boldsymbol{\beta_1},\ldots, \boldsymbol{\beta_n})\\
      & +(1-\delta_{\alpha^j,\beta_0^j})\sum_{\boldsymbol{\beta \in\mathcal{I}_j}}\ C_{\boldsymbol{\alpha}-\boldsymbol{\Delta},{n-1}}(\boldsymbol{\beta_0-\Delta+\beta,\underbrace{\beta_1,\ldots, \beta_n}_{\cancel{\beta}}})\binom{\boldsymbol{\beta_0-\Delta+\beta}}{\boldsymbol{\beta_0}}.
\end{split}
\end{equation}
To find the form of the $C_{\boldsymbol{\alpha},n}$, substitute the form \eqref{eq:lemma_1_coeff} for the coefficients on the RHS of \eqref{eq:coeff_n_gen}. Noting that,\begin{align}
        &\binom{\boldsymbol{\alpha-\Delta}}{\boldsymbol{\beta_0-\Delta,\beta_1,\ldots,\beta_n}}=\frac{\beta_0^j}{\alpha^j}\binom{\boldsymbol{\alpha}}{\boldsymbol{\beta_0,\beta_1,\ldots,\beta_n}} \text{ when } \beta_0^j>0,\\
        &\binom{\boldsymbol{\alpha-\Delta}}{\boldsymbol{\beta_0-\Delta+\beta,\underbrace{\beta_1,\ldots, \beta_n}_{\cancel{\beta}}}}\binom{\boldsymbol{\beta_0-\Delta+\beta}}{\beta_0}=\frac{\beta^j}{\alpha^j}\binom{\boldsymbol{\alpha}}{\boldsymbol{\beta_0,\beta_1,\ldots,\beta_n}},
\end{align}
we obtain,
\begin{equation}
\begin{split}
   &C_{\boldsymbol{\alpha},n}(\boldsymbol{\beta_0},\ldots, \boldsymbol{\beta_n})\\
     =&\frac{1}{\alpha^j}\binom{\boldsymbol{\alpha}}{\boldsymbol{\beta_0,\beta_1,\ldots,\beta_n}}\bigg[\frac{(1-\delta_{0,\beta_0^j})\beta_0^j}{\varpi(\boldsymbol{\beta_1,\ldots,\beta_n})}+(1-\delta_{\alpha^j,\beta_0^j})\sum_{\boldsymbol{\beta \in\mathcal{I}_j}}\ \frac{\beta^j}{\varpi(\underbrace{\boldsymbol{\beta_1,\ldots, \beta_n}}_{\cancel{\boldsymbol{\beta}}})}\bigg]\\
     =&\frac{1}{\alpha^j}\binom{\boldsymbol{\alpha}}{\boldsymbol{\beta_0,\beta_1,\ldots,\beta_n}}\frac{1}{\varpi(\boldsymbol{\beta_1,\ldots,\beta_n})}\bigg[(1-\delta_{0,\beta_0^j})\beta_0^j+(1-\delta_{\alpha^j,\beta_0^j})\sum_{\boldsymbol{\beta \in\mathcal{I}_j}}\ m(\boldsymbol{\beta}) \cdot \beta^j\bigg],\\
\end{split}
\end{equation}
where in the second line we used the relation $\varpi(\boldsymbol{\beta_1,\ldots,\beta_n})=\varpi(\underbrace{\boldsymbol{\beta_1,\ldots, \beta_n}}_{\cancel{\boldsymbol{\beta}}})m(\boldsymbol{\beta})$. The result follows from noting that the sum of terms in the square brackets is equal to $\alpha^j$ for all choices of $\beta_0^j$.

Finally, set $n=\alpha$ in \eqref{eq:second term}. Note that in this case, $\boldsymbol{\beta_0}=\boldsymbol{0}$ and $\mathcal{I}_j=\{\boldsymbol{\Delta}\}$, so we obtain,
\begin{equation}
 \begin{split}
     C_{\boldsymbol{\alpha},{\alpha}}(\boldsymbol{0,\beta_1,\ldots, \beta_{\alpha}})=C_{\boldsymbol{\alpha}-\boldsymbol{\Delta},{\alpha-1}}(\boldsymbol{0,\underbrace{\beta_1,\ldots, \beta_{\alpha}}_{\cancel{\Delta}}}).
\end{split}
\end{equation}
Note that when $n=\alpha$, the coefficients \eqref{eq:lemma_1_coeff} simplify  to $C_{\boldsymbol{\alpha},\alpha}(\boldsymbol{0},\boldsymbol{\beta_1,\ldots,\beta_{\alpha}})=1$ for all $\boldsymbol{\alpha}$. Verifying the base case completes the proof.
\end{proof}
\section{The causal set reduction formula}\label{sec:reduction}

In this section we present two theorems: the causal set reduction formula (Theorem \ref{th:lsz}) and its special case for distinct particles (Theorem \ref{th:lsz_noncollinear}). We begin with a definition.
\begin{definition}[The symbol $\mathcal{G}^{\boldsymbol{\alpha}}_{\mu_{1}\cdots\mu_l\rightarrow \lambda_{1}\cdots\lambda_k}$]\label{def:G}
    Given a pair of integers $k,l>0$ and a list $\boldsymbol{\alpha}$ with $|k-l|\leq\alpha\leq k+l$, the symbol $\mathcal{G}^{\boldsymbol{\alpha}}_{\mu_{1}\cdots\mu_l\rightarrow \lambda_{1}\cdots\lambda_k}$ is defined as,
    \begin{equation}\begin{split}
\mathcal{G}^{\boldsymbol{\alpha}}_{\mu_{1}\cdots\mu_l\rightarrow \lambda_{1}\cdots\lambda_k}:=\sum_{\substack{S_{\mu}\subseteq \{\mu_{1},\ldots,\mu_l\} \\ |S_{\mu}|=\frac{l+k-\alpha}{2}}}\ \sum_{\substack{S_{\lambda}\subseteq \{\lambda_{1},\ldots,\lambda_k\} \\ |S_{\lambda}|=\frac{l+k-\alpha}{2}}}\ &\sum_{\Upsilon:S_{\mu}\rightarrow S_{\lambda}}\prod_{\mu_i\in S_{\mu}} \frac{\delta_{\mu_i,\Upsilon(\mu_i)}}{\mu_i}\\&\sum_{\psi}\prod_{\mu_i\in S^c_{\mu} }v^{\mu_i}_{\psi(\mu_i)}\prod_{\lambda_j\in S^c_{\lambda}} \bar{v}^{\lambda_j}_{\psi(\lambda_j)},
\end{split}
    \end{equation}
where $\Upsilon:S_{\mu}\rightarrow S_{\lambda}$ is a bijection and $\psi$ is a map $\psi:S_{\mu}^c\cup S_{\lambda}^c \rightarrow\{1,\ldots,N\}$ in which the preimage of $i$ has cardinality $\alpha^i$ for all $i\in\{1,\ldots,N\}$.
\end{definition}

\begin{theorem}[The causal set reduction formula]\label{th:lsz}
    \begin{equation}
\langle \lambda_1\ldots \lambda_k|F|\mu_1\ldots \mu_l\rangle=\sum_{\substack{\boldsymbol{\alpha}\\ |k-l|\leq\alpha\leq k+l}}I_{\boldsymbol{\alpha}}(f)\  \mathcal{G}^{\boldsymbol{\alpha}}_{\mu_{1}\cdots\mu_l\rightarrow \lambda_{1}\cdots\lambda_k}.
\end{equation}
\end{theorem}

Before we prove the theorem, it is worth noting that Lemma \ref{lemma1} affords our reduction formula a double interpretation. It is a path integral expression for matrix elements---the missing piece from the formalism to date---and it is also a relation between S-matrix elements and correlation functions. The correlation functions in this case are of operators of the form $\partial V$ (not of an ordered product of $\phi$'s). In Section \ref{sec:DS}, we will gain a third equality by way of Theorem \ref{thm5.2}.

\begin{proof}[Proof of Theorem \ref{th:lsz}]
Note that it follows from the definitions in Section \ref{sec:prem} that, 
\begin{equation}\label{particle_states_equality}
    \langle \lambda_1\cdots\lambda_k|F|\mu_1\cdots\mu_l\rangle=  \frac{\langle \lambda_1\cdots\lambda_k|SF|\mu_1\cdots\mu_l\rangle_{\text{free}}}{\langle S \rangle_{\text{free}}}.
\end{equation}
To evaluate the RHS of \eqref{particle_states_equality}, we begin by enumerating all the arrangements in which every free particle state is contracted with a free field operator in $SF$.
To start, consider the case where all particle states are contracted into fields in $F$. The sum of these contributions is given by,
\begin{equation}\label{eq:0406eve}
\sum_{\substack{\boldsymbol{\alpha}\\\alpha=l+k}} \mathcal{G}^{\boldsymbol{\alpha}}_{\mu_{1}\cdots\mu_l\rightarrow \lambda_{1}\cdots\lambda_k}\frac{\langle S\partial^{\boldsymbol{\alpha}}F\rangle_{\text{free}}}{\langle S\rangle_{\text{free}}}=  \sum_{\substack{\boldsymbol{\alpha}\\\alpha=l+k}} \mathcal{G}^{\boldsymbol{\alpha}}_{\mu_{1}\cdots\mu_l\rightarrow \lambda_{1}\cdots\lambda_k}\langle \partial^{\boldsymbol{\alpha}}F\rangle.
\end{equation}
Each $\boldsymbol{\alpha}$ specifies the spacetime positions of the fields to be contracted with and the $\mathcal{G}^{\boldsymbol{\alpha}}_{\mu_{1}\cdots\mu_l\rightarrow \lambda_{1}\cdots\lambda_k}$ symbol encodes the different permutations in which the mode functions can be contracted into the specified fields. Both $\boldsymbol{\alpha}$ and $\mathcal{G}^{\boldsymbol{\alpha}}_{\mu_{1}\cdots\mu_l\rightarrow \lambda_{1}\cdots\lambda_k}$ treat coincident field operators as indistinguishable. The constant factor coming from the derivative ensures that we obtain the correct counting for distinguishable fields. 

Now consider the remaining contributions where at least one particle state is contracted into $S$. To do this, expand $S$ in powers of $V$ and consider powers greater than 0,

\begin{equation}\label{eq:0406}\begin{split}
\sum_{r=1}^{\infty}\frac{(-i)^r}{r!}\frac{\langle \mu_{1}\cdots\mu_l|FV^r|\lambda_{1}\cdots\lambda_k\rangle_{\text{free}}}{\langle S\rangle_{\text{free}}}.\\
\end{split} 
\end{equation}
In this case we encode the fields to be contracted into by a collection of lists $\boldsymbol{\beta_0},\ldots,\boldsymbol{\beta_n}$ with $n\leq r$, where $\boldsymbol{\beta_0}$ encodes the fields in $F$ that are contracted into and each of $\boldsymbol{\beta_1},\ldots,\boldsymbol{\beta_n}$ is a non-vanishing list that specifies fields in some separate copy of $V$. As above, apply the associated derivatives to remove the contracted fields from their respective operators,
\begin{equation}\label{eq:theorem1:mid}\frac{\big\langle V^{r-n} \partial^{\boldsymbol{\beta_0}}F\partial^{\boldsymbol{\beta_1}}V\cdots\partial^{\boldsymbol{\beta_n}} V \big\rangle_{\text{free}}}{\langle S\rangle_{\text{free}}}.\\
\end{equation}
We now need to count the number of ways that \eqref{eq:theorem1:mid} can come about. There are $\binom{r}{n}$ ways of choosing which of the $V$'s the derivatives will act on, and then there are $\frac{n!}{\varpi(p)}$ ways of distributing the derivatives over the chosen $V$'s, with $p=\{\boldsymbol{\beta_1},\ldots,\boldsymbol{\beta_n}\}$. Once this is fixed, there are $\binom{\boldsymbol{\alpha}}{\boldsymbol{\beta_0},\boldsymbol{\beta_1},\ldots,\boldsymbol{\beta_n}}$  ways of distributing the mode functions between the operators. The final factor needed to treat coincident fields as distinguishable is again provided by the derivatives. Putting this together, we find that \eqref{eq:theorem1:mid} should appear with a coefficient of, \begin{equation}
    \frac{r!}{(r-n)!}C_{\boldsymbol{\alpha},n}(\boldsymbol{\beta_0},\ldots,\boldsymbol{\beta_n}),
\end{equation}
where the $C_{\boldsymbol{\alpha},n}$ are as defined in \eqref{eq:lemma_1_coeff}.

Next, sum these contributions together and find that \eqref{eq:0406} can be rewritten as,
\begin{equation}\begin{split}
   \sum_{r=1}^{\infty}(-i)^r\sum_{\substack{\boldsymbol{\alpha}\\\alpha=l+k}}\mathcal{G}^{\boldsymbol{\alpha}}\sum_{n=1}^{min\{\alpha,r\}} \frac{1}{(r-n)!}\sum_{\substack{\boldsymbol{\beta_0}< \boldsymbol{\alpha} \\ \alpha-\beta_0\geq n}}\hspace{-0.5mm}&\sum_{\substack{P_n(\boldsymbol{\alpha}-\boldsymbol{\beta_0})}}\hspace{-2.5mm}C_{\boldsymbol{\alpha},n}(\boldsymbol{\beta_0},\ldots,\boldsymbol{\beta_n})\ \\
   &\ \ \ \ \ \frac{\big\langle V^{r-n} \partial^{\boldsymbol{\beta_0}}F\partial^{\boldsymbol{\beta_1}}V\cdots\partial^{\boldsymbol{\beta_n}} V \big\rangle_{\text{free}}}{\langle S\rangle_{\text{free}}}.
\end{split} 
\end{equation}
Change the order of summation using,
\begin{equation}
\sum_{r=1}^{\infty}\sum_{\substack{\boldsymbol{\alpha}\\\alpha=l+k}}\sum_{n=1}^{min\{\alpha,r\}} = \sum_{\substack{\boldsymbol{\alpha}\\\alpha=l+k}}\sum_{n=1}^{\alpha} \sum_{r=n}^{\infty},
\end{equation} apply the coordinate transformation $r-n\rightarrow r$ and resum $S$ inside the expectation value. This rids us of the $r$ sum and of the stray factors of $V$ on which no derivatives act. Replace the ratio of free correlators with an interacting one via,

\begin{equation}
\frac{\big\langle S\partial^{\boldsymbol{\beta_0}}F\partial^{\boldsymbol{\beta_1}}V\cdots\partial^{\boldsymbol{\beta_n}} V \big\rangle_{\text{free}}}{\langle S\rangle_{\text{free}}}=\langle\partial^{\boldsymbol{\beta_0}}F\partial^{\boldsymbol{\beta_1}}V\cdots\partial^{\boldsymbol{\beta_n}} V\rangle,
\end{equation}
and add \eqref{eq:0406eve} to the result to obtain,
 \begin{equation}\begin{split}\label{eq0506}
\sum_{\substack{\boldsymbol{\alpha}\\ \alpha=l+k}} \mathcal{G}^{\boldsymbol{\alpha}}_{\mu_{1}\cdots\mu_l\rightarrow \lambda_{1}\cdots\lambda_k}\Bigg[\big\langle \partial^{\boldsymbol{\alpha}}F\big\rangle+\sum_{n=1}^\alpha(-i)^n\hspace{-2.5mm}\sum_{\substack{\boldsymbol{\beta_0}< \boldsymbol{\alpha} \\ \alpha-\beta_0\geq n}}\hspace{-0.5mm}\sum_{\substack{P_n(\boldsymbol{\alpha}-\boldsymbol{\beta_0})}}\hspace{-2.5mm}&C_{\boldsymbol{\alpha},n}(\boldsymbol{\beta_0},\ldots,\boldsymbol{\beta_n})\ \\
&\big\langle \partial^{\boldsymbol{\beta_0}}F\partial^{\boldsymbol{\beta_1}}V\cdots\partial^{\boldsymbol{\beta_n}} V\big\rangle\Bigg]\\
&\hspace{-87mm}=\sum_{\substack{\boldsymbol{\alpha}\\ \alpha=l+k}}I_{\boldsymbol{\alpha}}(f)\  \mathcal{G}^{\boldsymbol{\alpha}}_{\mu_{1}\cdots\mu_l\rightarrow \lambda_{1}\cdots\lambda_k},
\end{split}
\end{equation}
where the last line follows directly from Lemma \ref{lemma1}.

Equation \eqref{eq0506} gives the contribution to the matrix element in which there are no freely propagating particles. To obtain the remaining contributions, return to \eqref{particle_states_equality}, contract $n$ in-particles with $n$ out-particles and repeat the steps above to enumerate the number of ways of contracting the remaining particle states with the fields. The result, after summing over all the ways of contracting $n$ in-particles with $n$ out-particles is,
 \begin{equation}\begin{split}\label{proof_final}
\sum_{\substack{\boldsymbol{\alpha}\\ \alpha=l+k-2n}}I_{\boldsymbol{\alpha}}(f)\  \mathcal{G}^{\boldsymbol{\alpha}}_{\mu_{1}\cdots\mu_l\rightarrow \lambda_{1}\cdots\lambda_k}.
\end{split}
\end{equation}
Summing \eqref{proof_final} over $n\in\{0,1,\ldots,\min\{l,k\}\}$ completes the proof.
\end{proof}

When no in-particle is identical to an out-particle (cf. the non-collinear case in the continuum) Theorem \ref{th:lsz} simplifies and we obtain the following theorem as a direct corollary.
    \begin{theorem}[The causal set reduction formula for distinct particles]\label{th:lsz_noncollinear}\begin{equation}
\langle \lambda_1\ldots \lambda_k|F|\mu_1\ldots \mu_l\rangle=\sum_{\substack{\boldsymbol{\alpha}\\ \alpha= k+l}}\bigg[\sum_{\psi_{\boldsymbol{\alpha}}}\prod_{i=1}^l v^{\mu_i}_{\psi_{\boldsymbol{\alpha}}(\mu_i)}\prod_{j=1}^{k} \bar{v}^{\lambda_j}_{\psi_{\boldsymbol{\alpha}}(\lambda_j)}\bigg]I_{\boldsymbol{\alpha}}(f),
\end{equation}
 where the sum is over all maps $\psi_{\boldsymbol{\alpha}}:\{\mu_{1},\cdots,\mu_l, \lambda_{1},\cdots,\lambda_k\}\rightarrow\{1,\ldots,N\}$ in which the preimage of $i$ has cardinality $\alpha^i$ for all $i\in\{1,\ldots,N\}$.
\end{theorem}

As a consistency check, note that S-matrix elements are obtained by setting $F=1$. In that case, Lemma 1 simplifies to,
\begin{equation}\label{eq:lemma1_f=1}
    \begin{split}
I_{\boldsymbol{\alpha}}(1)=&\sum_{n=1}^\alpha(-i)^n\sum_{\substack{P_n(\boldsymbol{\alpha})}}\frac{1}{\varpi(p)} \binom{\boldsymbol{\alpha}}{\boldsymbol{\beta_1},\ldots,\boldsymbol{\beta_n}} \ \big\langle \partial^{\boldsymbol{\beta_1}}V\cdots\partial^{\boldsymbol{\beta_n}} V\big\rangle.\\
\end{split}
    \end{equation}
Setting $V=0$ trivially recovers the free theory, and similarly setting $k=l=0$ recovers the vacuum in-out correlator of the interacting theory.

\section{Quantum equations of motion}\label{sec:DS}
In this section we present Theorem \ref{thm5.2}, giving a tower of relations satisfied by correlation functions on a causal set. We begin with some notation and a lemma.

For any symmetric matrix $M$ and list $\boldsymbol{\beta}$ satisfying $\beta=2$, we use the following shorthand to denote a matrix element of $M$:
\begin{equation}
        M_{\boldsymbol{\beta}}=\begin{cases}
M_{ij}& \text{ if }\beta^i=\beta^j=1\\
M_{ii} & \text{ if } \beta^i=2.\\
\end{cases}
    \end{equation}
Consider a list $\boldsymbol{\alpha}$ with $\alpha=2n$ for some $n\in\mathbb{N}$. We will be interested in the partitions of $\boldsymbol{\alpha}$ into $n$ lists of size 2, \textit{i.e.} a partition $\boldsymbol{\beta_1,\ldots ,\beta_n}$ of $\boldsymbol{\alpha}$ with $\beta_i=2$ for all $i$. We denote the set of all such partitions of $\boldsymbol{\alpha}$ as $\mathscr{P}^{(2)}(\boldsymbol{\alpha})$.

 \begin{lemma}\label{lemma lemma 2}
Given some list $\boldsymbol{\alpha}\in\mathbb{N}^N$ with $\alpha>1$, \begin{equation}
    \begin{split}
\partial^{\boldsymbol{\alpha}}D=D \Bigg[ \prod_{m=1}^N  (&-\mathcal{E}_m\hspace{-0.5mm} \cdot \hspace{-0.5mm}\xi)^{\alpha^m}\\
+&\sum_{n=1}^{\lfloor \frac{\alpha}{2}\rfloor}\sum_{\substack{\boldsymbol{\beta_0}< \boldsymbol{\alpha} \\ \alpha-\beta_0= 2n}}\hspace{-0.5mm}\sum_{\substack{\mathscr{P}^{(2)}(\boldsymbol{\alpha}-\boldsymbol{\beta_0})}}\hspace{-5mm}C_{\boldsymbol{\alpha},n}(\boldsymbol{\beta_0},\ldots,\boldsymbol{\beta_n})\ \prod_{r=1}^n (-\mathcal{E}_{\boldsymbol{\beta_r}})\prod_{m=1}^N (-\mathcal{E}_m\hspace{-0.5mm} \cdot \hspace{-0.5mm}\xi)^{\beta_0^m}\Bigg],\\
\end{split}
\end{equation}
where the $C_{\boldsymbol{\alpha},n}$ are as defined in \eqref{eq:lemma_1_coeff} and where it is understood that the sum $\sum_{n=1}^{\lfloor \frac{\alpha}{2}\rfloor}$ is empty when ${\lfloor \frac{\alpha}{2}\rfloor}<1$.
\end{lemma}

\begin{proof}
    Assume true for all lists less than or equal to $\boldsymbol{\alpha}$. Define $\boldsymbol{\Delta}$ via $\Delta^i=\delta_{ij}$ for all $i=1,\ldots,N$.

    Take the $\partial^{\boldsymbol{\Delta}}\equiv\frac{\partial}{\partial\xi_j}$ derivative of the lemma statement to obtain,

\begin{equation}
    \begin{split}
\partial^{\boldsymbol{\alpha+\Delta}}D=&D \Bigg[ (-\mathcal{E}_j\hspace{-0.5mm} \cdot \hspace{-0.5mm}\xi) \prod_{m=1}^N  (-\mathcal{E}_m\hspace{-0.5mm} \cdot \hspace{-0.5mm}\xi)^{\alpha^m}+\sum_{t=1}^N \alpha^t(-\mathcal{E}_{tj})(-\mathcal{E}_t\hspace{-0.5mm} \cdot \hspace{-0.5mm}\xi)^{\alpha^t-1}\prod_{m\not=t}  (-\mathcal{E}_m\hspace{-0.5mm} \cdot \hspace{-0.5mm}\xi)^{\alpha^m}\\\\
+&\sum_{n=1}^{\lfloor \frac{\alpha}{2}\rfloor}\sum_{\substack{\boldsymbol{\beta_0}< \boldsymbol{\alpha} \\ \alpha-\beta_0= 2n}}\hspace{-0.5mm}\sum_{\substack{\mathscr{P}^{(2)}(\boldsymbol{\alpha}-\boldsymbol{\beta_0})}}\hspace{-5mm}C_{\boldsymbol{\alpha},n}(\boldsymbol{\beta_0},\ldots,\boldsymbol{\beta_n})\prod_{r=1}^n (-\mathcal{E}_{\boldsymbol{\beta_r}})\\ &\bigg((-\mathcal{E}_j\hspace{-0.5mm} \cdot \hspace{-0.5mm}\xi)  \prod_{m=1}^N (-\mathcal{E}_m\hspace{-0.5mm} \cdot \hspace{-0.5mm}\xi)^{\beta_0^m}+\sum_{t=1}^N \beta_0^t(-\mathcal{E}_{tj})(-\mathcal{E}_t\hspace{-0.5mm} \cdot \hspace{-0.5mm}\xi)^{\beta_0^t-1}\prod_{m\not=t}  (-\mathcal{E}_m\hspace{-0.5mm} \cdot \hspace{-0.5mm}\xi)^{\beta_0^m}\bigg)\Bigg].
\end{split}
\end{equation}
    The first term in the square brackets can be rewritten as $\prod_{m=1}^N  (-\mathcal{E}_m\hspace{-0.5mm} \cdot \hspace{-0.5mm}\xi)^{(\alpha+\Delta)^m}$ so gives the first term in the lemma statement. 
In the second term, the sum over $t$ is equivalent to a sum over $\boldsymbol{\beta_0}<\boldsymbol{\alpha}$ with the constraint $\alpha-\beta_0=1$. Therefore it can be written as,
\begin{equation}\label{eq:09062}
    D \sum_{\substack{\boldsymbol{\beta_0}< \boldsymbol{\alpha}+\boldsymbol{\Delta}\\\beta_0^j\leq\alpha^j \\ \alpha+\Delta-\beta_0= 2}}\binom{\boldsymbol{\alpha}}{\boldsymbol{\beta_0}}\big(-\mathcal{E}_{\boldsymbol{\alpha+\Delta}-\boldsymbol{\beta_0}}\big)\prod_{m=1}^N(-\mathcal{E}_m\hspace{-0.5mm} \cdot \hspace{-0.5mm}\xi)^{\beta_0^m},
\end{equation}
where we have used the fact that when $\alpha-\beta_0=1$ then 
$\prod_{t=1}^N (1-\delta_{\alpha^t,\beta_0^t})\alpha^t=\binom{\boldsymbol{\alpha}}{\boldsymbol{\beta_0}}$.
The third term can be written as,
\begin{equation}\label{eq:09061}
    \begin{split}
D \sum_{n=1}^{\lfloor \frac{\alpha}{2}\rfloor}\sum_{\substack{\boldsymbol{\beta_0}< \boldsymbol{\alpha}+\boldsymbol{\Delta} \\ \beta_0^j\geq 1\\ \alpha+\Delta-\beta_0= 2n}}\hspace{-0.5mm}\sum_{\substack{\mathscr{P}^{(2)}(\boldsymbol{\alpha+\Delta}-\boldsymbol{\beta_0})}}\hspace{-5mm}C_{\boldsymbol{\alpha},n}(\boldsymbol{\beta_0-\Delta,\beta_1},\ldots,\boldsymbol{\beta_n})\prod_{r=1}^n (-\mathcal{E}_{\boldsymbol{\beta_r}}) \prod_{m=1}^N (-\mathcal{E}_m\hspace{-0.5mm} \cdot \hspace{-0.5mm}\xi)^{\beta_0^m}.
\end{split}
\end{equation}
By summing \eqref{eq:09062} with the $n=1$ term of \eqref{eq:09061} and then changing variables $\boldsymbol{\alpha}\rightarrow\boldsymbol{\alpha-\Delta}$, we obtain the recurrence relation \eqref{eq:recursion_n_1}.

Similarly to the second term, the final term can be written as,
    \begin{equation}\label{eq09063}
    \begin{split}
D \sum_{n=1}^{\lfloor \frac{\alpha}{2}\rfloor}\sum_{\substack{\boldsymbol{\beta_0}< \boldsymbol{\alpha} \\ \alpha-\beta_0= 2n}}\sum_{\substack{\boldsymbol{\gamma}< \boldsymbol{\beta_0} \\ \beta_0-\gamma=1}}\hspace{-0.5mm}\sum_{\substack{\mathscr{P}^{(2)}(\boldsymbol{\alpha}-\boldsymbol{\beta_0})}}\hspace{-5mm}C_{\boldsymbol{\alpha},n}(\boldsymbol{\beta_0},\ldots,\boldsymbol{\beta_n})\binom{\boldsymbol{\beta_0}}{\boldsymbol{\gamma}}(-\mathcal{E}_{\boldsymbol{\beta_0+\Delta}-\boldsymbol{\gamma}}\big)\prod_{r=1}^{n} (-\mathcal{E}_{\boldsymbol{\beta_r}})\prod_{m=1}^N  (-\mathcal{E}_m\hspace{-0.5mm} \cdot \hspace{-0.5mm}\xi)^{\gamma^m}.
\end{split}
\end{equation}
Relabel $\boldsymbol{\beta_0}\leftrightarrow\boldsymbol{\gamma}$, rearrange the sums using,
\begin{equation}
\sum_{\substack{\boldsymbol{\gamma}< \boldsymbol{\alpha} \\ \alpha-\gamma= 2n}}\sum_{\substack{\boldsymbol{\beta_0}< \boldsymbol{\gamma} \\ \gamma-\beta_0=1}}=\sum_{\substack{\boldsymbol{\beta_0}< \boldsymbol{\alpha} \\ \alpha-\beta_0=2n+1}}\sum_{\substack{\boldsymbol{\beta_0}<\boldsymbol{\gamma}< \boldsymbol{\alpha} \\ \gamma-\beta_0=1}},
\end{equation}
and eliminate $\boldsymbol{\gamma}$ using $\boldsymbol{\beta_{n+1}}:=\boldsymbol{\gamma-\beta_0+\Delta}$. Note that for any list $\boldsymbol{\varpi}$ with $\varpi=2n+1$ and any function $g(\boldsymbol{\beta};\boldsymbol{\beta_{1},\ldots,\beta_{n}})$, we have, \begin{equation}
\sum_{\substack{\boldsymbol{\Delta}<\boldsymbol{\beta}< \boldsymbol{\varpi+\Delta} \\ \beta=2}}
\sum_{\substack{\mathscr{P}^{(2)}(\boldsymbol{\varpi+\Delta-\beta})}} g(\boldsymbol{\beta};\boldsymbol{\beta_{1},\ldots,\beta_{n}})=
\sum_{\substack{\mathscr{P}^{(2)}(\boldsymbol{\varpi+\Delta})}} \sum_{\boldsymbol{\beta}\in \mathcal{I}_j} g(\boldsymbol{\beta};\underbrace{\boldsymbol{\beta_1},\ldots, \boldsymbol{\beta_{n+1}}}_{\cancel{\boldsymbol{\beta}}})
\end{equation}
where $\mathcal{I}_j$ is as defined in \eqref{def set I}, and use this to rewrite equation \eqref{eq09063} as,
\begin{equation}\label{eq10061}
    \begin{split}
D \Bigg[&\sum_{n=2}^{\lfloor \frac{\alpha}{2}\rfloor}\sum_{\substack{\boldsymbol{\beta_0}< \boldsymbol{\alpha+\Delta} \\ \beta_0^j\leq\alpha^j\\\alpha+\Delta-\beta_0= 2n}}\sum_{\substack{\mathscr{P}^{(2)}(\boldsymbol{\alpha}+\Delta-\boldsymbol{\beta_0})}}\sum_{\boldsymbol{\beta}\in\mathcal{I}_j}C_{\boldsymbol{\alpha},n-1}(\boldsymbol{\beta_0+\beta-\Delta},\underbrace{\boldsymbol{\beta_1},\ldots, \boldsymbol{\beta_{n}}}_{\cancel{\boldsymbol{\beta}}})\\
&\hspace{50mm}\binom{\boldsymbol{\beta_0+\beta-\Delta}}{\boldsymbol{\beta_0}}\prod_{r=1}^{n} (-\mathcal{E}_{\boldsymbol{\beta_r}})\prod_{m=1}^N  (-\mathcal{E}_m\hspace{-0.5mm} \cdot \hspace{-0.5mm}\xi)^{\beta_0^m}\\
+&\sum_{\substack{\boldsymbol{\beta_0}< \boldsymbol{\alpha+\Delta} \\ \alpha+\Delta-\beta_0= 2\lfloor \frac{\alpha}{2}\rfloor+2}}\sum_{\substack{\mathscr{P}^{(2)}(\boldsymbol{\alpha+\Delta}-\boldsymbol{\beta_0})}}\sum_{\boldsymbol{\beta}\in\mathcal{I}_j}C_{\boldsymbol{\alpha},n}(\boldsymbol{\beta_0+\beta-\Delta},\underbrace{\boldsymbol{\beta_1},\ldots, \boldsymbol{\beta_{n+1}}}_{\cancel{\boldsymbol{\beta}}})\\
&\hspace{50mm}\binom{\boldsymbol{\beta_0+\beta-\Delta}}{\boldsymbol{\beta_0}}\prod_{r=1}^{n+1} (-\mathcal{E}_{\boldsymbol{\beta_r}})\prod_{m=1}^N  (-\mathcal{E}_m\hspace{-0.5mm} \cdot \hspace{-0.5mm}\xi)^{\beta_0^m}\Bigg].
\end{split}
\end{equation}
Summing the first term of \eqref{eq10061} with the $n>2$ terms of \eqref{eq:09061} and relabeling $\boldsymbol{\alpha+\Delta}\rightarrow\boldsymbol{\alpha}$, we find that the recursion relation for $2\leq n\leq \lfloor \frac{\alpha}{2}\rfloor$ is given by \eqref{eq:coeff_n_gen}. When $\alpha$ is even then this completes the proof, since $\lfloor \frac{\alpha+\Delta}{2}\rfloor=\lfloor \frac{\alpha}{2}\rfloor$ and the second term in \eqref{eq10061} vanishes (since it is a sum over $\boldsymbol{\beta_0}$ with $\beta_0<0$). When $\alpha$ is odd, we get an additional relation from the second term of \eqref{eq10061},
\begin{equation}
\begin{split}
  C_{\boldsymbol{\alpha+\Delta}, \frac{\alpha+1}{2}}(\boldsymbol{0,\boldsymbol{\beta_1},\ldots,\beta_\frac{\alpha+1}{2}})=&\sum_{\boldsymbol{\beta}\in\mathcal{I}_j} C_{\boldsymbol{\alpha},\frac{\alpha-1}{2}}(\boldsymbol{\beta-\Delta},\underbrace{\boldsymbol{\beta_1},\ldots, \boldsymbol{\beta_{\frac{\alpha+1}{2}}}}_{\cancel{\boldsymbol{\beta}}}).
  \end{split}
\end{equation}
Plugging in \eqref{eq:lemma_1_coeff} on the RHS gives,
\begin{equation}
\begin{split}
  C_{\boldsymbol{\alpha+\Delta}, \frac{\alpha+1}{2}}(\boldsymbol{0,\boldsymbol{\beta_1},\ldots,\beta_\frac{\alpha+1}{2}})=\frac{1}{\varpi(\boldsymbol{\boldsymbol{\beta_1},\ldots,\beta_\frac{\alpha+1}{2}})}&\binom{\boldsymbol{\alpha+\Delta}}{\boldsymbol{\beta_1,\ldots,\beta_\frac{\alpha+1}{2}}}\\
  &\frac{1}{\alpha^j+1}\sum_{\boldsymbol{\beta}\in\mathcal{I}_j}m(\boldsymbol{\beta})\boldsymbol{\beta}!.
  \end{split}
\end{equation}
The result follows from noting that since $\boldsymbol{\beta}\in p\in\mathscr{P}^{(2)}(\boldsymbol{\alpha}+\Delta)$ then $\boldsymbol{\beta!}=\beta^j$ and hence $\sum_{\boldsymbol{\beta}\in\mathcal{I}_j}m(\boldsymbol{\beta})\boldsymbol{\beta}!=\alpha^j+1$.
\end{proof}

The following theorem is now immediate.

   \begin{theorem}[Quantum equations of motion on a causal set]\label{thm5.2}
Given a list $\boldsymbol{\alpha}\in\mathbb{N}^N$ and polynomial operator $F=F(\phi_1,\ldots,\phi_N)$, and assuming that \newline $\xi\cdot Re[\mathcal{E}] \cdot\xi\geq 0$ for all $\xi\in\mathbb{R}^N$, then,    \begin{equation}\label{eq:theorem2}
    \begin{split}
&\big\langle F\prod_{k=1}^N  (\mathcal{E}_k\hspace{-0.5mm} \cdot \hspace{-0.5mm}\phi)^{\alpha^k}\big\rangle+\sum_{n=1}^{\lfloor \frac{\alpha}{2}\rfloor}\sum_{\substack{\boldsymbol{\beta_0}< \boldsymbol{\alpha} \\ \alpha-\beta_0= 2n}}\hspace{-0.5mm}\sum_{\substack{\mathscr{P}^{(2)}(\boldsymbol{\alpha}-\boldsymbol{\beta_0})}}C_{\boldsymbol{\alpha},n}(\boldsymbol{\beta_0},\ldots,\boldsymbol{\beta_n})\ \prod_{k=1}^n (-\mathcal{E}_{\boldsymbol{\beta_i}}) \big\langle F\prod_{k=1}^N (\mathcal{E}_k\hspace{-0.5mm} \cdot \hspace{-0.5mm}\phi)^{\beta_0^k}\big\rangle \\&=\big\langle \partial^{\boldsymbol{\alpha}}F\big\rangle+\sum_{n=1}^\alpha(-i)^n\hspace{-2.5mm}\sum_{\substack{\boldsymbol{\beta_0}< \boldsymbol{\alpha} \\ \alpha-\beta_0\geq n}}\hspace{-0.5mm}\sum_{\substack{P_n(\boldsymbol{\alpha}-\boldsymbol{\beta_0})}}\hspace{-2.5mm}C_{\boldsymbol{\alpha},n}(\boldsymbol{\beta_0},\ldots,\boldsymbol{\beta_n})\ \big\langle \partial^{\boldsymbol{\beta_0}}F\partial^{\boldsymbol{\beta_1}}V\cdots\partial^{\boldsymbol{\beta_n}} V\big\rangle,\\
\end{split}
\end{equation}
where the $C_{\boldsymbol{\alpha},n}$ are as defined in \eqref{eq:lemma_1_coeff}.

\end{theorem}
\begin{proof}
    Consider $I_{\boldsymbol{\alpha}}(f)$ as defined in \eqref{eq:integral1}. Use Lemma \ref{lemma lemma 2} to evaluate $\partial^{\boldsymbol{\alpha}}D$ inside the integrand and apply \eqref{eq:pthint} to evaluate the integral as,
         \begin{equation}\label{eq58}
    \begin{split}
I_{\boldsymbol{\alpha}}(f)=&\big\langle F\prod_{k=1}^N  (\mathcal{E}_k\hspace{-0.5mm} \cdot \hspace{-0.5mm}\phi)^{\alpha^k}\big\rangle\\
&+\sum_{n=1}^{\lfloor \frac{\alpha}{2}\rfloor}\sum_{\substack{\boldsymbol{\beta_0}< \boldsymbol{\alpha} \\ \alpha-\beta_0= 2n}}\hspace{-0.5mm}\sum_{\substack{\mathscr{P}^{(2)}(\boldsymbol{\alpha}-\boldsymbol{\beta_0})}}C_{\boldsymbol{\alpha},n}(\boldsymbol{\beta_0},\ldots,\boldsymbol{\beta_n})\ \prod_{k=1}^n (-\mathcal{E}_{\boldsymbol{\beta_i}}) \big\langle F\prod_{k=1}^N (\mathcal{E}_k\hspace{-0.5mm} \cdot \hspace{-0.5mm}\phi)^{\beta_0^k}\big\rangle.
\end{split}
\end{equation} The theorem then follows directly from Lemma \eqref{lemma1}.
\end{proof}

Choosing $\boldsymbol{\alpha}$ with $\alpha=1$ and $F=\phi_{a_1}\ldots \phi_{a_n}$, we obtain the direct analogue of the continuum's Schwinger-Dyson equations, since in that case the second term on the LHS of \eqref{eq:theorem2} vanishes, $\langle \partial^{\boldsymbol{\alpha}}F\rangle$ gives the contact terms and the final term gives the interaction term. For instance, setting $\boldsymbol{\alpha}=(1,0,\ldots,0)$ and $F=\phi_a$ Theorem \ref{thm5.2} gives,
\begin{equation}
  \mathcal{E}_{1b} \langle \phi_a\phi_b\rangle =\delta_{1a} -i\langle \phi_a\frac{\partial V}{\partial\xi_1}\rangle,
\end{equation}
while the comparable Schwinger-Dyson equation in a theory with Lagrangian $\mathcal{L}=\frac{1}{2}(\partial_{\mu}\phi)^2-\frac{1}{2}m^2\phi^2-V[\phi(x)]$ is,
\begin{equation}\label{SDeq}
  i(\Box_x+m^2)\langle\Omega| \phi(x)\phi(y)|\Omega\rangle =\delta^{(4)}(x-y)-i\Big\langle \Omega\Big| \phi(y)\frac{\delta V}{\delta \phi(x)}\Big|\Omega\Big\rangle.
\end{equation}
This observation further strengthen the notion that $\mathcal{E}$ is in effect the causal set equation of motion operator (up to a factor of $i$, depending on convention). When $\alpha>1$, we can therefore interpret Theorem \ref{thm5.2} as the closed form for repeated applications of the equation of motion operator on a correlation function---the ingredient we needed to bring the causal set reduction formula into a form that resembles the continuum LSZ formula in position space. Thus, Theorem \ref{thm5.2} affords a new interpretation to the reduction formulas of Section \ref{sec:reduction}: S-matrix elements can be obtained from causally-ordered correlators of field operators through the removal of external legs by application of the equation of motion operator $\mathcal{E}$, similarly to the standard interpretation in the continuum. We note two differences from the continuum. First, on a causal set we have an additional interpretation: a matrix element is equal to an alternating sum of field correlators---with all external legs intact---since our $\mathcal{E}$ is a complex matrix, not a differential operator. Second, on the causal set there is an additional term that is responsible for the cancelation of diagrams that contain free propagation, \textit{i.e.} the second term on the LHS of \eqref{eq:theorem2}. We illustrate this cancellation with an explicit example below and discuss the issue of free propagation further in the next Section.

\begin{example}
    Let $\boldsymbol{\alpha}=(4,0,\ldots,0)$. Note that by \eqref{eq58} we have,
    \begin{equation}\label{eq:example}
        I_{\boldsymbol{\alpha}}(1)=\langle (\mathcal{E}_1\cdot\phi)^4\rangle-6\mathcal{E}_{11}\langle (\mathcal{E}_1\cdot\phi)^2\rangle +3\mathcal{E}_{11}^2.
    \end{equation}
Expand the first term on the RHS as,
\begin{spreadlines}{1cm}
\begin{flalign*}
&\mathcal{E}_{1a}\mathcal{E}_{1b}\mathcal{E}_{1c}\mathcal{E}_{1d} \langle \phi_a\phi_b\phi_c\phi_d\rangle=\mathcal{E}_{1a}\mathcal{E}_{1b}\mathcal{E}_{1c}\mathcal{E}_{1d}\Bigg[\bigg( \begin{tikzpicture}[baseline=(current bounding box.center), decoration={markings,
    mark= at position 0.6 with {\arrow{Straight Barb[ length=1mm, width=1mm]}}}
] 
  \tikzstyle{every node}=[circle, draw, fill=black,inner sep=0pt, minimum width=3pt]
  \node (n1) at (9,0) [label= left: $a$]{};
  \node (n2) at (10,0) [label= right: $b$]{};
   \node (n3) at (9,-0.5) [label= left: $c$]{};
  \node (n4) at (10,-0.5) [label= right: $d$]{};
    \draw [](n1) -- (n2)  {} ;
     \draw [](n3) -- (n4)  {} ;
\end{tikzpicture}+2 \ permutations\bigg)\\
&\vspace{20mm}+\bigg(\begin{tikzpicture}[baseline=(current bounding box.center), decoration={markings,
    mark= at position 0.6 with {\arrow{Straight Barb[ length=1mm, width=1mm]}}}
] 
  \tikzstyle{every node}=[circle, draw, fill=black,inner sep=0pt, minimum width=3pt]
  \node (n1) at (9,-0) [label= left: $a$]{};
  \node (n2) at (10,0) [label= right: $d$]{};
   \node (n3) at (9,-0.5) [label= left: $b$]{};
  \node (n4) at (10,-0.5) [label= right: $c$]{};
     \draw [](n4) -- (10,-1.25)  {} ;
       \node (n5) at  (9,-1.25)  [label= below: $-iV$]{};
        \node (n6) at  (10,-1.25)  [label= below: $-iV$]{};
    \draw [](n1) -- (n2)  {} ;
     \draw [](n3) -- (n5)  {} ;
\end{tikzpicture}+5 \ permutations\bigg)+ \bigg(\begin{tikzpicture}[baseline=(current bounding box.center), decoration={markings,
    mark= at position 0.6 with {\arrow{Straight Barb[ length=1mm, width=1mm]}}}
] 
  \tikzstyle{every node}=[circle, draw, fill=black,inner sep=0pt, minimum width=3pt]
  \node (n1) at (9,0) [label= left: $a$]{};
  \node (n2) at (10,0) [label= right: $b$]{};
   \node (n3) at (9,-0.5) [label= left: $c$]{};
  \node (n4) at (10,-0.5) [label= right: $d$]{};
  \node(n5) at (9.5,-1.25) [label= below: $-iV$]{};
    \draw [](n1) -- (n2)  {} ;
     \draw [](n3) -- (n5)  {} ;
       \draw [](n4) -- (n5)  {} ;
\end{tikzpicture}+5 \ permutations\bigg)+X_{abcd}\Bigg] \\
=&3\mathcal{E}_{11}^2+6\mathcal{E}_{11}(-i)^2\langle \partial_1V\partial_1V\rangle+6\mathcal{E}_{11}(-i)\langle\partial_1^2V\rangle+\mathcal{E}_{1a}\mathcal{E}_{1b}\mathcal{E}_{1c}\mathcal{E}_{1d}X_{abcd}
\end{flalign*}
\end{spreadlines}
    where $X_{abcd}\subset\langle \phi_a\phi_b\phi_c\phi_d\rangle$ denotes the sum of diagrams where no pair of external fields is contracted. Similarly expanding the second term on the RHS of \eqref{eq:example} we obtain,
    \begin{equation*}
    -6\mathcal{E}_{11}\langle (\mathcal{E}_1\cdot\phi)^2\rangle=-6\mathcal{E}_{11}^2-6\mathcal{E}_{11}(-i)^2\langle \partial_1V\partial_1V\rangle-6\mathcal{E}_{11}(-i)\langle\partial_1^2V\rangle.
    \end{equation*}
    Plugging these back into $\eqref{eq:example}$ we find that $I_{\boldsymbol{\alpha}}(1)=\mathcal{E}_{1a}\mathcal{E}_{1b}\mathcal{E}_{1c}\mathcal{E}_{1d}X_{abcd}$ and hence $I_{\boldsymbol{\alpha}}(1)$ contains no terms with free propagation.
\end{example}

\section{Comparison with the continuum}\label{sec:discussion}
We conclude with a short summary and a discussion of how our work makes contact with the continuum.

We gave a reduction formula for matrix elements on a causal set (Theorem \ref{th:lsz}) and its special case for distinct particles (Theorem \ref{th:lsz_noncollinear}). Our proofs relied on an expansion of the integral $I_{\boldsymbol{\alpha}}(f)$ (Definition~\ref{definitionI}) in terms of correlators built from interaction potentials (Lemma \ref{lemma1}). A key assumption has been that the equation of motion matrix $\mathcal{E}$ (defined as the inverse of the Feynman propagator) exists and is positive definite. In practice, in the causal set literature today the Feynman propagator is obtained via the Sorkin-Johnston prescription, starting from a freely-chosen retarded propagator. Whether $\mathcal{E}$ exists and has or does not have a given property depends on that initial choice of retarded propagator and on the particular causal set background the QFT lives on. The literature has various prescriptions through which one might choose a retarded propagator for a given background, but which is the ``right'' propagator remains an open question (see \cite{Hinrichsen:2026vmv} for a recent proposal). Our work suggests that a good choice of retarded propagator is one that gives rise to a positive definite $\mathcal{E}$, since this endows the theory with the relationships that one would expect between correlators and matrix elements and enables us to make contact with the continuum. Note that in the continuum, the standard path integral prescription fails in the presence of zero modes and one needs to manage the resulting divergences to obtain a path integral description of the theory (for instance, see \cite{PhysRevD.110.116023} and references therein). Studying how one should manage the divergences in \eqref{eq:pthint} in the presence of zero modes (that is, when $\mathcal{E}$ does not exist) is a worthwhile direction for future research.

The structure of the integral $I_{\boldsymbol{\alpha}}(f)$ and its application to obtaining a reduction formula has the following analogue in the continuum. Consider a simple scalar field theory on Minkowski with generating functional,
\begin{align}
&Z[J]=\int \mathcal{D}\phi\  \exp\Bigg[i\int d^dx\ \bigg( \frac{1}{2}(\partial\phi)^2-\frac{1}{2}m^2\phi^2-V(\phi)+ J(x)\phi(x)\bigg)\Bigg].
\end{align}
Define $\phi_0$ via $(\Box+m^2)\phi_0=J$, complete the square and change variables $\phi\rightarrow \phi+\phi_0$ to obtain,
\begin{align}
&Z[J]=e^{\frac{i}{2}\int d^dx J(x)\frac{1}{\Box+m^2}J(x)}\int \mathcal{D}\phi \ e^{i\int d^dx\ \big(\frac{1}{2}(\partial\phi)^2-\frac{1}{2}m^2\phi^2-V(\phi+\phi_0)\big)}.
\end{align}
Correlators are obtained by taking functional derivatives with respect to $J$. The term outside the integral gives free propagators, while the term inside the integral gives all diagrams that do not contain free propagation. Now recall the textbook LSZ reduction formula for non-collinear particles on Minkowski,
\begin{equation}\label{eq:textbook}
\begin{split}
    &_{out}\langle p_1,\ldots,p_n|q_1,\ldots,q_m\rangle_{in}\\
      &=\int \prod^m_{l=1}d^4x_l \frac{ie^{-iq_l\cdot x_l}(\Box_{x_l}+m^2)}{\sqrt{Z}}\prod^{n+m}_{k=m+1}d^4x_k \frac{ie^{ip_k\cdot x_k}(\Box_{x_k}+m^2)}{\sqrt{Z}}G(x_1,\ldots,x_{m+n}),
\end{split}
  \end{equation}
where $G(x_1,\ldots,x_{m+n})$ is the position space correlation function in the interacting vacuum. One can show that every diagram in $G$ that contains free propagation (that is, a contraction of two external fields) vanishes once plugged into \eqref{eq:textbook} with the upshot that we can replace $G$ with, \begin{equation}\label{cont_gen_int}
 (-i)^{n+m}  \frac{\delta^{n+m} Z_{int}[J]}{\delta J(x_1)\cdots\delta J(x_{n+m})}\bigg|_{J=0},
\end{equation} where we define $Z_{int}$ and $D_{con}$ (the continuum analogue of the causal set decoherence functional $D$, which we will use below) via,
\begin{equation}
    Z_{int}[J]\equiv \int \mathcal{D}\phi \ e^{i\int d^dx\ \big(\frac{1}{2}(\partial\phi)^2-\frac{1}{2}m^2\phi^2-V(\phi+\phi_0)\big)} \equiv \int \mathcal{D}\phi \ D_{con} \ e^{-i\int d^dx\ V(\phi+\phi_0)}.
\end{equation}
The action of the Klein-Gordon operators on \eqref{cont_gen_int} (omitting the factors of $-i$) evaluates to,
\begin{equation}\label{eq:off shell continuum PI int}\begin{split}
&\prod_{j=1}^{m+n} (\Box_j+m^2) \ \frac{\delta^{m+n} Z_{int}[J]}{\delta J(x_1)\cdots\delta J(x_{m+n})}\bigg|_{J=0}\\
&=(-1)^{m+n}\hspace{-2mm}\int \mathcal{D}\phi \ e^{-i\int d^dx V(\phi+\phi_0)}  \frac{\delta^{m+n} D_{con}[\phi]}{\delta \phi(x_1)\cdots\delta \phi(x_{m+n})},
\end{split}
\end{equation}
where we changed variables $J\rightarrow \phi_0$, brought the $\phi_0$ derivatives inside the integral where they can be replaced with $\phi$ derivatives acting solely on the interaction part (since $D_{con}$ is independent of $\phi_0$) and finally integrated by parts. The last line is the continuum analogue of our $I_{\boldsymbol{\alpha}}(1)$. Putting the factors of $-i$ back in and plugging the result into \eqref{eq:textbook}, we find that all factors of $i$ cancel out and we obtain a form analogous to the causal set reduction formula, barring the factors of wavefunction remormlisation which to date do not have a causal set analogue. Note that Section \ref{sec:DS} suggests another way of bringing \eqref{eq:textbook} closer in form to the causal set reduction formula: replace $ \prod_{i=1}^{n+m}(\Box_{x_i}+m^2)\ G(x_1,\ldots,x_{m+n})$ in \eqref{eq:textbook} with ``higher order'' Schwinger-Dyson equations to express matrix elements as a sum of correlators of interaction potentials (cf. \eqref{SDeq}).

The fact that diagrams with free propagation give only vanishing contributions to \eqref{eq:textbook} is a consequence of the Fock space structure,\footnote{It is not a consequence of momentum conservation: these diagrams still vanish if we set some incoming momentum equal to some outgoing momentum in \eqref{eq:textbook}. In the continuum, the contribution of free propagation to collinear scattering comes from an additional boundary term; on the causal set this contribution has been encoded in Definition \ref{def:G}.} \textit{i.e.} of our choice of vacuum and the resulting commutation relations between the ladder operators. This property is preserved in our causal set setup. Our choice of vacuum---the causal set Sorkin-Johnston vacuum---is encoded in the decoherence functional $D$, and therefore in $I_{\boldsymbol{\alpha}}(f)$. Indeed, the form of $I_{\boldsymbol{\alpha}}(f)$ given in Lemma \ref{lemma1} makes it explicit that there is no free propagation in Theorem \ref{th:lsz_noncollinear}, since every mode function is attached to some interaction vertex $\partial V$. To retain this property on a curved background in the continuum, one needs to take care to choose an appropriate ``comoving'' vacuum or else encounter manifest particle production (of the sort that accounts for the CMB power-spectrum from inflation). Our work suggests that the continuum Sorkin-Johnston vacuum on de Sitter \cite{Surya:2018byh} may be a good choice of a comoving vacuum that will retain vanishing contributions from free theory particle production, much like the Bunch-Davies state \cite{Melville:2023kgd}.

\paragraph{Acknowledgment:} The author is indebted to Scott Melville and Andrew Tolley for sharing their time and expertise during the preparation of this work.

\bibliographystyle{unsrt}
\bibliography{bio}
\end{document}